\documentclass[a4paper,11pt]{article}

\usepackage{graphicx}
\usepackage{caption,color}
\usepackage{amsmath,amssymb,amsthm,enumerate}
\usepackage{fullpage}
\usepackage[algo2e,ruled,vlined]{algorithm2e}
\usepackage{verbatim}
\usepackage{hyperref}
\hypersetup{colorlinks=true}
\hypersetup{linkcolor=black,anchorcolor=black,citecolor=black,urlcolor=black}
\hypersetup{pdftitle=Competitive Algorithms for Generalized $k$-Server in Uniform Metrics}
\hypersetup{pdfauthor={Nikhil Bansal} {Marek Elias} {Grigorios Koumoutsos} {Jesper Nederlof}}
\hypersetup{pdfkeywords={generalized k-server problem} {online algorithms}
	{competitive analysis}}

\newtheorem{theorem}{Theorem}[section]
\newtheorem{definition}[theorem]{Definition}

\newtheorem{observation}[theorem]{Observation}
\newtheorem{lemma}[theorem]{Lemma}
\newtheorem{corol}[theorem]{Corollary}
\newtheorem{claim}[theorem]{Claim}

\DeclareMathOperator{\ALG}{ALG}
\DeclareMathOperator{\ADV}{ADV}
\DeclareMathOperator{\OPT}{OPT}

\DeclareMathOperator{\WFA}{WFA}

\DeclareMathOperator{\cost}{cost}

\newcommand{\cF}{\ensuremath{\mathcal{F}}}
\newcommand{\cM}{\ensuremath{\mathcal{M}}}
\newcommand{\R}{\mathbb{R}}

\newcounter{note}

\begin{document}

\title{Competitive Algorithms for Generalized $k$-Server in Uniform Metrics%
\thanks{This work was supported by NWO grant 639.022.211, ERC consolidator
grant 617951, and NWO Veni project 639.021.438}}

\author{
Nikhil Bansal\thanks{TU Eindhoven, Netherlands.
\texttt{\{n.bansal,m.elias,g.koumoutsos,j.nederlof\}@tue.nl}}
\qquad
Marek Eli\'a\v{s}\footnotemark[2]
\qquad
Grigorios Koumoutsos\footnotemark[2]\\[1.2ex]
\qquad
Jesper Nederlof\footnotemark[2]\\[1.2ex]
TU Eindhoven, Netherlands}
\maketitle

\begin{abstract}
The generalized $k$-server problem is a far-reaching extension of the $k$-server problem with several applications. 
Here, each server $s_i$ lies in its own metric space $M_i$. A request is a $k$-tuple $r = (r_1,r_2,\dotsc,r_k)$ and to serve it, we need to move some server $s_i$ to the point $r_i \in M_i$, and the goal is to minimize the total distance traveled by the servers.
Despite much work, no $f(k)$-competitive algorithm is known for the problem for $k > 2$ servers, even for special cases such as uniform metrics and lines.

Here, we consider the problem in uniform metrics and give the first $f(k)$-competitive algorithms for general $k$. In particular, we obtain deterministic and randomized algorithms with competitive ratio $O(k 2^k)$  and $O(k^3 \log k)$ respectively. Our deterministic bound is based on a novel application of the polynomial method to online algorithms, and essentially matches the long-known lower bound of $2^k-1$.
We also give a $2^{2^{O(k)}}$-competitive deterministic algorithm for weighted uniform metrics, which also essentially matches the recent doubly exponential lower bound for the problem.

\end{abstract}

\thispagestyle{empty}
\newpage
\setcounter{page}{1}
\section{Introduction}

The $k$-server problem was proposed by Manasse et al.~\cite{MMS90} as a far-reaching generalization of many online problems, and its study has led to various remarkable developments~\cite{BEY98,Kou09,KP95,BBMN15}. In this problem, we are given $k$-servers $s_1,\dotsc,s_k$ located at points of a metric space $M$. At each time step a request arrives at some point of $M$ and must be served by moving some server there. The goal is to minimize the total distance traveled by the servers.

Koutsoupias and Taylor \cite{KT04} introduced a substantial generalization of the $k$-server problem, called the \textit{generalized $k$-server problem}. Here, each server $s_i$ lies in its own metric space $M_i$, with its own distance function $d_i$. A request is a $k$-tuple $r = (r_1,r_2,\dotsc,r_k)$ and must be served by moving some server $s_i$ to the point $r_i \in M_i$. 
Note that the standard $k$-server problem corresponds to the special case when all the metrics are identical, $M_1 = \ldots =M_k=M$,
and the requests are of the form $(r,r,\ldots,r)$, i.e.,~the $k$-tuple is identical in each coordinate.

The generalized $k$-server problem can model a rich class of online problems, for which the techniques developed for the standard $k$-server problem do not apply, see e.g.~\cite{KT04}. For that reason, it is widely believed that a deeper understanding of this problem should lead to powerful new techniques for designing online algorithms \cite{KT04,Sit14}. According to Koutsoupias and Taylor \cite{KT04}, this problem ``may act as a stepping stone towards building a robust (and less ad hoc) theory of online computation''. 

\subsection{Previous Work}

\paragraph{The $k$-server problem.} The $k$-server problem has been extensively studied (an excellent reference is~\cite{BEY98}).
The initial work focused on special metrics such as uniform metrics and lines, and optimum competitive ratios were obtained in many cases~\cite{CKPV91,CL91,KP96}. A particularly interesting case is that of uniform metrics, which corresponds to the very well-studied \textit{paging} problem, where tight $k$-competitive deterministic~\cite{ST85} and $O(\log k )$-competitive randomized algorithms~\cite{FKL+91,MS91,ACN00} are known.

For general metrics,
Koutsoupias and Papadimitriou~\cite{KP95} showed in a breakthrough result that the Work Function Algorithm ($\WFA$) is $(2k-1)$-competitive in any metric space. This essentially matches the lower bound of $k$ for any deterministic algorithm~\cite{MMS90}. More recently, a $\textrm{polylog}(k,n)$ randomized competitive algorithm was obtained~\cite{BBMN15} where $n$ in the number of points in $M$.

\paragraph{The generalized $k$-server problem.} This problem is much less understood. In their seminal paper, Koutsoupias and Taylor~\cite{KT04} studied the special case where $k=2$ and both the metrics $M_1$ and $M_2$ are lines. This is called CNN problem and it has attracted a lot of attention \cite{AG10,Chr03,IY04,iw01}. They showed that, even for this special case, many successful $k$-server algorithms or their natural generalizations are not competitive. 

\paragraph{Lower Bounds:} 
For uniform metrics, Koutsoupias and Taylor~\cite{KT04} showed that even when each $M_i$ contains $n=2$ points, the competitive ratio is at least $2^k - 1$. For general metrics, the best known lower bound is $2^{2^{\Omega(k)}}$~\cite{BEK17}, and comes from the weighted $k$-server problem (the weighted variant of the standard $k$-server problem). This problem corresponds to generalized-$k$-server where the metric spaces are scaled copies of each other, i.e.~$M_i = w_i M$ for some fixed $M$, and the requests have the form $(r,\ldots,r)$.

\paragraph{Upper Bounds:}
Despite considerable efforts, competitive algorithms\footnote{Those with competitive ratio $f(k)$ that only depends on $k$. Note that an $n^k-1$ competitive algorithm follows trivially, as the problem can be viewed as Metrical Service System (MSS) on $n^k$ states, where $n=\max_{i=1}^k |M_i|$.}
  are known only for the case of $k=2$ servers~\cite{SS06,Sit14,SSP03}. In a breakthrough result, Sitters and Stougie~\cite{SS06} obtained a $O(1)$-competitive algorithm for $k=2$ in any metric space.
Recently, Sitters~\cite{Sit14} showed that the generalized WFA is also $O(1)$-competitive for $k=2$ by a careful and subtle analysis of the structure of work functions. Despite this progress, no $f(k)$-competitive algorithms are known for $k>2$, even for special cases such as uniform metrics and lines.

\subsection{Our Results}

We consider the generalized $k$-server problem on uniform metrics and obtain the
first $f(k)$-competitive algorithms for general $k$, whose competitive ratios
almost match the known lower bounds.

Perhaps surprisingly, there turn out to be two very different settings for uniform metrics:
\begin{enumerate}
\item
When all the metric spaces $M_1,\ldots,M_k$ are uniform (possibly with different number of points) with identical pairwise distance, say $1$. We call this the {\em uniform metric} case.
\item When the metric spaces $M_i$ are all uniform, but have different scales, i.e.~all pairwise distances in $M_i$ are $w_i$. We call this the {\em weighted uniform metric} case.
 \end{enumerate}

Our first result is the following. 

\begin{theorem}
\label{thm:alg_det}
There is a $O(k 2^k)$ competitive deterministic algorithm for the generalized $k$-server problem in the
uniform metric case.
\end{theorem} 
This almost matches the $2^k-1$ lower bound due to \cite{KT04} (we describe this instructive and simple lower bound instance in the Appendix for completeness).

The proof of Theorem \ref{thm:alg_det} is based on a general combinatorial argument about how the set of feasible states evolves
as requests arrive. Specifically, we divide the execution of the algorithm in phases, and consider the beginning of a phase when all the MSS states are feasible (e.g. the cost is $0$ and not $\infty$). As requests arrive, the set of states that remain valid for all requests during this phase can only reduce. In particular, for this problem we show that any sequence of requests that causes the feasible state space to strictly reduce at each step, can have length at most $2^k$ until all states becomes infeasible.

Interestingly, this argument is based on a novel application of the polynomial or the rank method from linear algebra \cite{Jukna,Matousek,Guth}.
While the rank method has led to some spectacular recent successes in combinatorics and computer science \cite{KakeyaDvir,Capset}, we are not aware of any previous applications to online algorithms. 
We feel our approach could be useful for other online problems that can be modeled as Metrical Service Systems by analyzing the combinatorial structure in a similar way.

\vspace{2mm}

Next, we consider randomized algorithms against oblivious adversaries. 
\begin{theorem}
\label{thm:alg_rand}
There is a randomized algorithm for the generalized $k$-server problem on
uniform metrics with competitive ratio $O(k^3 \log k)$.
\end{theorem}

The rank method above does not seem to be useful in the randomized setting as it only bounds the number of requests until the set of feasible states becomes empty, and does not give any structural information about how the set of states evolves over time.
As we observe in Section \ref{sec:alg_rand}, a $o(2^k)$ guarantee cannot be obtained without using such structural information. So we explore the properties of this evolution more carefully and use it to design the randomized algorithm in Theorem \ref{thm:alg_rand}.
 
In the Appendix, we also give a related lower bound. In particular, we note that an $\Omega(k/\ln^2 k)$ lower bound on the competitive ratio of any randomized algorithm follows directly by combining the  lower bound instance of \cite{KT04} with the results of \cite{BBM06}. 

\vspace{2mm}

Finally, we consider the weighted uniform metric case.

\begin{theorem}
\label{thm:weight_comp}
There is a $2^{2^{k+3}}$ competitive algorithm for generalized $k$-server on weighted uniform metrics.  
\end{theorem}
Theorem \ref{thm:weight_comp} follows by observing that a natural modification of an algorithm due to Fiat and Ricklin~\cite{FR94} for weighted $k$-server on uniform metrics also works for the more general generalized $k$-server setting.
Our proof is essentially the same as that of \cite{FR94}, with some arguments streamlined and an improved competitive ratio\footnote{It was first pointed out to us by Chiplunkar~\cite{Chip-pc} that the competitive ratio $2^{2^{4k}}$ claimed in~\cite{FR94} can be improved to $2^{2^{k+O(1)}}$.}.
Finally, note that the $2^{2^{\Omega(k)}}$ lower bound \cite{BEK17} for weighted $k$-server on uniform metrics implies that Theorem \ref{thm:weight_comp} is essentially optimal.

\section{Deterministic algorithm for uniform metrics}
\label{sec:alg_det}
In this section we prove Theorem \ref{thm:alg_det}. Recall that each $M_i$ is the uniform metric with unit distance. We assume that all metrics have $n = \max_{i=1}^{k} |M_i|$ points (if for some metric $|M_i| < n$, we can add some extra points that are never requested).
We use $[n]$ to denote $\{1,\ldots,n\}$. As the requests are arbitrary $k$-tuples and each metric $M_i$ is uniform, we can relabel the points arbitrarily and hence assume that the set of points in each $M_i$ is $[n]$.
At any time $t$, the state of an algorithm can be described by the $k$-tuple $q^t = (q^t_1,\ldots,q^t_k)$ where for each $i\in [k]$, $q^t_i \in [n]$ denotes the location of server $i$. Let $r^t = (r^t_1,\ldots,r^t_k)$ denote the request vector at time $t$. We need to find a state with the following property:
\begin{definition}
	A state $q^t$ {\em satisfies (or is feasible for)} the request $r^t$ if $q^t_i = r^t_i$ for some $i\in [k]$.
\end{definition}
Moreover, if the state changes from $q^t$ to $q^{t+1}$, the algorithm pays the Hamming distance \[d(q^{t+1},q^t) = |\{i: q^{t+1}_i \neq q^t_i\}|,\] between $q^{t}$ and $q^{t+1}$.

We describe a generic algorithm below that works in phases in Algorithm~\ref{alg:det1}. We will show that during each phase the offline moves at least once and hence pays at least $1$, while the online algorithm changes its state at most $2^k$ times and hence pays at most $k 2^k$ as the Hamming distance between any two states is at most $k$. This will be sufficient as the offline optimum will need to change its state at least once as no state satisfies all requests, and it follows that our algorithm pays at most $(c^{\ast}+1)k2^k$, where $c^{\ast}$ denotes the optimal cost. Here the $+1$ accounts for the last (possible unfinished) phase.

\begin{algorithm2e}
	\DontPrintSemicolon
	If a phase begins, the algorithm starts in some arbitrary $q^1$.\;
	At each time $t$ when a request $r^t$ arrives do the following.\;
	\eIf{\normalfont the current state $q^t$ does not satisfy the current request $r^t$}{
		\eIf{\normalfont there exists a state $q$ that satisfies all requests $r^1,\ldots,r^{t}$}{
			Set $q^{t+1}=q$. 
		}
		{
			Set $q^{t+1}$ to be an arbitrary location satisfying (only) $r^t$.\;
			End the current phase. \;			
		}
	}{
			Set $q^{t+1}=q^t$.\;
	}
	\caption{A deterministic $O(k2^k)$ competitive algorithm.}
	\label{alg:det1}
\end{algorithm2e}
	
We call this algorithm {\em generic} as it can pick any arbitrary point $q$ as long as it is feasible for $r^1,\ldots,r^t$. 
Note that this algorithm captures a wide variety of natural algorithms including (variants) of the Work Function Algorithm.

Fix some phase that we wish to analyze, and let $\ell$ denote its length.
Without loss of generality, we can assume that $r^t$ always causes $q^t$ to move (removing such requests does not reduce the online cost, and can only help the offline adversary). So the online algorithm moves exactly $\ell$ times.
Moreover, the adversary must move at least once during the phase as no location exists that satisfies all the requests $r^1,\ldots,r^{\ell}$ that arrive during the phase.

It suffices to show the following.
\begin{theorem}
For any phase as defined above, its length satisfies $\ell \leq 2^k$.
\end{theorem}
\begin{proof}
We use the rank method. 
Let $x=(x_1,\ldots,x_k), y=(y_1,\ldots,y_k)$ be points in $\R^k$, and
consider the $2k$-variate degree $k$ polynomial $p:\R^{2k} \rightarrow \R$, 
\[ p(x,y) :=  \prod_{i\in [k]}  (x_i-y_i). \]
The key property of $p$ is that a state $q \in [n]^k$ satisfies a request $r \in [n]^k$ iff  $p(q,r)=0$.

We now construct a matrix $M$ that captures the dynamics of the online algorithm during a phase.
Let $M \in \R^{\ell \times \ell}$ be an $\ell \times \ell$ matrix,
where columns correspond to the states and rows to the requests,
with entries $M[t,t'] = p(q^t,r^{t'})$, i.e.,~the $[t,t']$ entry of 
$M$ corresponds to the evaluation of $p$ on $q^t$ and $r^{t'}$
\begin{claim}
$M$ is an upper triangular matrix with non-zero diagonal.
\end{claim}
\begin{proof}
At any time $t=1,\ldots,\ell$, as the current state $q^t$ does not satisfy the request $r^t$, it must be that $p(q^t,r^t) \neq 0$.

On the other hand, for $t=2,\ldots,\ell$, the state $q^t$ was chosen such that it satisfied all the previous requests $t'$ for $t'<t$.
This gives that  
$M[t,t'] =0$ for $t'<t$ and hence all the entries below the diagonal are $0$.
\end{proof}

As the determinant of any upper-triangular matrix  is the product of its diagonal entries, this implies that $M$ has non-zero determinant and has full rank, $\textrm{rk}(M) =\ell$.

However, we can use the structure of $p$ to show that the rank of $M$ is at most $2^k$ in a fairly straight manner\footnote{Curiously, this particular rank upper bound was used in a previous work for answering a question the a completely different setting about the parameterized complexity of graph coloring parameterized by cutwidth~\cite{Geffen17}.}.
In particular, we  give an explicit factorization of $M$ as $M=AB$, where $A$ is $\ell \times 2^k$ matrix and
$M$ is a $2^k \times \ell$ matrix. 
Clearly, as any $m \times n$ matrix has rank at most $\min(m,n)$, both $A$ and
$B$ have rank at most $2^k$.
Moreover, as $\textrm{rk}(AB) \leq \min(\textrm{rk}(A),\textrm{rk}(B))$, this implies $\textrm{rk}(M) \leq 2^k$. It remains to  show the factorization.

Indeed, if we express $p(x,y)$ in terms of its $2^k$ monomials, we can write 
\[p(x,y) = \sum_{S \subseteq [k]} (-1)^{k-|S|} X_S Y_{[k]\setminus S}, \]
where $X_{S} = \prod_{i\in S} x_i$ with $X_{\emptyset}=1$, and $Y_{S}$ is defined analogously. 

Now, let $A$ be the $\ell \times 2^k$ matrix with rows indexed by time $t$ and
columns by subsets $S \in 2^{[k]}$, with the entries
\[ A[t,S] = q^t_S:= \prod_{i\in S} q^t_i.\]
Similarly, let $B$ be the $ 2^k \times \ell$ matrix with rows indexed by subsets
$S \in 2^{[k]}$ and columns indexed by time $t'$. We define 
\[B[S,t'] = (-1)^{k-|S|}\, r^{t'}_{[k]\setminus S}  := (-1)^{k-|S|} \prod_{i \in [k]\setminus S} r^{t'}_i.\]

Then, for any $t,t'\in [\ell]$,
\[
M [t,t'] = p(q^t,r^{t'}) = \sum_{S \subseteq [k]} (-1)^{k-|S|}\; q^t_S\, r^{t'}_{[k]\setminus S} = \sum_{S \subseteq [k]} A[t,S]B[S,t'] = (AB)[t,t'].\]
and hence $M=AB$ as claimed.
\end{proof}

We remark that an alternate way to view this result is that the length
of any request sequence that causes the set of feasible states
to strictly decrease at each step
can be at most $2^k$.

\section{Randomized algorithm for uniform metrics}
\label{sec:alg_rand}
A natural way to randomize the algorithm above would be to pick a state uniformly at random among all the 
 states that are feasible for all the requests thus far in the current phase.
The standard randomized uniform MTS analysis \cite{BLS92} implies that this online algorithm would move
$O(\log (n^k)) = O(k \log n)$ times. However, this guarantee is not useful if $n \gg \exp(\exp(k))$.

Perhaps surprisingly, even if we use the fact from Section \ref{sec:alg_det} that the set of feasible states can shrink at most $2^k$ times, this does not suffice to give a randomized $o(2^k)$ guarantee. 
Indeed, consider the algorithm that picks a random state among the feasible ones in the current phase.
If, at each step $t=1,\ldots,\ell$, half of the feasible states become infeasible (expect the last step when
all states become infeasible), then the algorithm must move with probability at
least $1/2$ at each step, and hence incur an expected $\Omega(\ell) = \Omega(2^k)$ cost during the phase.

So proving a better guarantee would require showing that the scenario above cannot happen. In particular, we need a more precise understanding of how
the set of feasible states evolves over time, rather than simply a bound on the number of requests in a phase.

To this end, in Lemmas~\ref{lem:hydra} and~\ref{lem:Ft_size} below, we impose some stronger subspace-like structure over the set of
feasible states. Then, we use this structure to design a variant of the natural randomized algorithm above, that directly works with these subspaces.

\paragraph{Spaces of configurations.} Let $U_i$ denote the set of points in $M_i$. We can think of $U_i=[n]$, but $U_i$ makes the notation clear. We call state in $\prod_{i=1}^k U_i = [n]^k$ a configuration. Here we slightly abuse notation by letting $\prod$ denote the generalized Cartesian product.
It will be useful to consider sets of configurations where some server locations are
fixed at some particular location. For a vector
$v\in \prod_{i=1}^k (U_i\cup \{*\})$, we define the space
\[S(v) := \left\{c\in \prod_{i=1}^k U_i \,\middle|\;
	c_i = v_i \: \forall i \text{ s.t. } v_i \neq *\right\}.\]
A coordinate $i$ with $v_i = *$ is called {\em free} and the corresponding server can
be located at an arbitrary point of $U_i$. The number of free coordinates in the
space $S(v)$ we call dimension and denote it with $\dim(S(v))$.

Let us consider a $d$-dimensional space $S$ and a request $r$ such that some
configuration $c \in S$ is not feasible for $r$.
Then, we claim that a vast majority of configurations from $S$ are
infeasible for $r$, as stated in the following lemma.
We denote $F(r)$ the set of configuration satisfying $r$.

\begin{lemma}\label{lem:hydra}
Let $S$ be a $d$-dimensional space and let $r$ be a request which
makes some configuration $c\in S$ infeasible.
Then, there exist $d$ subspaces $S_1, \dotsc, S_d$, each of dimension $d-1$,
such that we have $S \cap F(r) = S_1 \cup \dotsb \cup S_d$.
\end{lemma}

Note that if all the metric spaces $U_i$ contain $n$ points, then
$|S_i| = \frac1n |S|$ for each $i = 1, \dotsc, d$.

\begin{proof}
By reordering the coordinates, we can assume that the first $d$ coordinates of
$S$ are free and $S$ corresponds to the vector
$(*, \dotsc, *, s_{d+1}, \dotsc, s_k)$, for some
$s_{d+1}, \dotsc, s_k$.
Let $r=(r_1,\ldots,r_k)$.

Consider the  subspaces $S(v_1), \dotsc, S(v_d)$, where
\[ v_1 = (r_1, *, \dotsc, *, s_{d+1}, \dotsc, s_k),
	\dotsc, v_d = (*, \dotsc, *, r_d, s_{d+1}, \dotsc, s_k).
\]
Clearly, any configuration contained in  
$S(v_1) \cup \dotsc \cup S(v_d)$,
is feasible for $r$.
Conversely, as there exists $c\in S$ infeasible for $r$, we have
$s_i=c_i\neq r_i$ for each $i=d+1, \dotsc, k$.
This already implies that each configuration from $S$ feasible for $r$
must belong to $S(v_1) \cup \dotsc \cup S(v_d)$:
whenever $c'\in S$ is feasible for $r$, it needs to have
$c'_i = r_i$ for some $i\in \{1, \dotsc, d\}$ and therefore
$c'\in S(v_i)$.
\end{proof}

\paragraph{Spaces of feasible configurations.}
During each phase, we maintain a set $\cF^t$ of spaces containing
configurations which were feasible with respect to the requests
$r^1, \dotsc, r^t$.
In the beginning of the phase,
we set $\cF^1 = \{(r^1_1,*,\dotsc,*), \dotsc, (*,\dotsc,*,r^1_k)\}$,
and, at time $t$, we update it in the following way.
We remove all spaces of dimension $0$ whose single configuration is infeasible
w.r.t. $r^t$.
In addition, we replace each $S\in \cF^{t-1}$ of dimension $s > 0$ which
contains some infeasible configuration by $S_1, \dotsc, S_d$
according to the Lemma~\ref{lem:hydra}. The following observation follows
easily from Lemma~\ref{lem:hydra}.

\begin{observation}\label{lem:Ft_feas}
Let us consider a phase with requests $r^1, \dotsc, r^\ell$.
A configuration $c$ is feasible
with respect to the requests $r^1, \dotsc, r^t$ if and only if
$c$ belongs to some space in $\cF^t$.
\end{observation}

\paragraph{An alternative deterministic algorithm.}
Based on $\cF^t$, we can design an alternative deterministic algorithm that has a
competitive ratio of $3k!$. This is worse than Algorithm~\ref{alg:det1} but will
be very useful to obtain our randomized algorithm.
To serve a request at time $t$, it chooses some space $Q^t \in \cF^t$ and moves
to an arbitrary $q^t\in Q^t$.
Whenever $Q^{t-1}$ no more belongs to $\cF^t$,
it moves to another space $Q^t$ regardless whether $q^{t-1}$ stayed feasible or
not, see Algorithm~\ref{alg:det2} for details. While, this is not an optimal behaviour,
a primitive exploitation of the
structure of $\cF^t$ already gives a reasonably good algorithm.

\begin{algorithm2e}
\DontPrintSemicolon
at time $t$:\;
\ForEach(\tcp*[f]{update $\cF^t$ for $r^t$})%
	{\normalfont $S\in \cF^{t-1}$ containing some infeasible configuration}{
	replace $S$ by $S_1, \dotsc, S_d$ according to Lemma~\ref{lem:hydra}\;
}
\If(\tcp*[f]{start a new phase, if needed}){$\cF^t = \emptyset$}{
	Set $\cF^t = \{S((r^t_1,*,\dotsc,*)), \dotsc, S((*,\dotsc,*,r^t_k))\}$\;
}
\eIf(\tcp*[f]{serve the request}){$Q^{t-1}\in \cF^t$}{
	set $Q^t := Q^{t-1}$ and $q^{t} := q^{t-1}$\;
}{
	choose arbitrary $Q^t \in \cF^t$ and move to an arbitrary $q^t\in Q^t$\;
}
\caption{Alternative deterministic algorithm.}
\label{alg:det2}
\end{algorithm2e}

The following lemma bounds the maximum number of distinct spaces which can
appear in $\cF^t$ during one phase.
In fact, it already implies that the competitive ratio of
Algorithm~\ref{alg:det2} is at most
$k! \cdot \sum_{d=0}^{k-1} \frac1{d!} \leq 3 k!$.

\begin{lemma}\label{lem:Ft_size}
Let us consider a phase with requests $r^1, \dotsc, r^\ell$.
Then $\bigcup_{t=1}^{\ell} \cF^t$ contains at most
$k!/d!$ spaces of dimension $d$.
\end{lemma}
\begin{proof}
We proceed by induction on $d$.
In the beginning, we have $k = k!/(k-1)!$ spaces of dimension $k-1$ in
$\cF^1$ and, by Lemma~\ref{lem:hydra}, all spaces added later have strictly
lower dimension.

By the way $\cF^t$ is updated, each $(d-1)$-dimensional space is created
from some $d$-dimensional space already present in
$\bigcup_{t=1}^{\ell} \cF^t$.
By the inductive hypothesis, there could be at most
$k!/d!$ distinct $d$-dimensional spaces
and Lemma~\ref{lem:hydra} implies that each of
them creates at most $d$ distinct $(d-1)$-dimensional spaces.
Therefore, there can be at most $\frac{k!}{d!} d = \frac{k!}{(d-1)!}$ spaces of
dimension $d-1$ in $\bigcup_{t=1}^{\ell} \cF^t$.
\end{proof}

\paragraph{Randomized algorithm.}
Now we transform Algorithm~\ref{alg:det2} into a randomized one.
Let $m_t$ denote the largest dimension among all the spaces in $\cF^t$ and let
$\cM^t$ denote the set of spaces of dimension $m_t$ in $\cF^t$.

The algorithm works as follows:
Whenever moving, it picks a space $Q^t$ from $\cM^t$ uniformly at
random, and moves to some arbitrary $q^t \in Q^t$.
As the choice of $q^t$ is arbitrary, whenever some configuration from $Q^t$
becomes infeasible, the algorithm assumes that $q^t$ is infeasible as well\footnote{This is done to keep the calculations simple, as
the chance of $Q^t$ being removed from $\cF$ and $q^t$ staying feasible is
negligible when $k \ll n$.}.

\begin{algorithm2e}
\DontPrintSemicolon
at time $t$:\;
\ForEach(\tcp*[f]{update $\cF^t$ for $r^t$})%
	{\normalfont $S\in \cF^{t-1}$ containing some infeasible configuration}{
	replace $S$ by $S_1, \dotsc, S_d$ according to Lemma~\ref{lem:hydra}\;
}
\If(\tcp*[f]{start a new phase, if needed}){$\cF^t = \emptyset$}{
	Set $\cF^t = \{S((r^t_1,*,\dotsc,*)), \dotsc, S((*,\dotsc,*,r^t_k))\}$\;
}
\eIf(\tcp*[f]{serve the request}){$Q^{t-1}\in \cM^t$}{
	set $Q^t := Q^{t-1}$ and $q^{t} := q^{t-1}$\;
}{
	Choose a space $Q^t$ from $\cM^t$ uniformly at random\;
	Move to an arbitrary $q^t \in Q^t$\;
}
\caption{Randomized Algorithm for Uniform metrics.}
\label{alg:rand}
\end{algorithm2e}

At each time $t$, $\ALG$ is located at some configuration $q^t$ contained in
some space in $\cF^t$ which implies that its position is feasible with respect to
the current request $r^t$, see Lemma~\ref{lem:Ft_feas}.
Here is the key property about the state of $\ALG$.

\begin{lemma}\label{lem:prob_dist}
At each time $t$, the probability of $Q^t$ being equal to some fixed $S \in
\cM^t$ is $1/|\cM^t|$.
\end{lemma}
\begin{proof}
If $\ALG$ moved at time $t$, the statement follows trivially, since $Q^t$ was
chosen from $\cM^t$ uniformly at random.
So, let us condition on the event that $Q^t= Q^{t-1}$. 

Now, the algorithm does not change state if and only if $Q^{t-1} \in \cM^{t}$.
Moreover, in this case $m_t$ does not change, and $\cM^t \subset \cM^{t-1}$.
By induction, $Q^{t-1}$ is distributed uniformly within $\cM^{t-1}$,
and hence conditioned on $Q^{t-1} \in \cM^t$, $Q^{t}$ is uniformly distributed within $\cM^t$.
\end{proof}

\paragraph{Proof of Theorem~\ref{thm:alg_rand}.}
At the end of each phase (except possibly for the last unfinished phase), the set of feasible states 
$\cF^t = \emptyset$, and hence $\OPT$ must pay at least $1$ during each of those
phases.
Denoting $N$ the number of phases needed to serve the entire request sequence,
we have $\cost(\OPT) \geq (N-1)$.
On the other hand, the expected online cost is at most,
\[ E[\cost(\ALG)] \leq c(N-1) + c \leq c \cost(\OPT) + c, \]
where $c$ denotes the expected cost of $\ALG$ in one phase.
This implies that $\ALG$ is $c$-competitive, and strictly $2c$-competitive (as the offline must move at least once, if the online algorithm pays a non-zero cost).

Now we prove that $c$ is at most
$O(k^3 \log k)$.  To show this, we use a potential function
\[ \Phi(t) = H(|\cM^t|) + \sum_{d=0}^{m_t-1} H(k!/d!), \]
where $H(n)$ denotes the $n$th harmonic number.
As the beginning of the phase, 
$\Phi(1) \leq k H(k!) \leq k(\log k!+1) = O(k^2 \log k)$ as $|\cM^{1}|\leq k!$ and $m_{1} \leq k-1$. Moreover the phase ends whenever
$\Phi(t)$ decreases to 0.
Therefore, it is enough to show that, at each time $t$,
the expected cost incurred by the algorithm is at most
$k$ times the decrease of the potential.
We distinguish two cases.

If $m_t = m_{t-1}$, let us denote $b = |\cM^{t-1}| - |\cM^t|$.
If $b > 0$, the potential decreases, and its change can be bounded as
\[
\Delta\Phi \leq H(|\cM^t|) - H(|\cM^{t-1}|)
	= - \frac1{|\cM^t|+1} - \frac1{|\cM^t|+2}
		- \dotsb - \frac1{|\cM^t|+b}
	\leq -b \cdot \frac1{|\cM^{t-1}|}.
\]
On the other hand, the expected cost of $\ALG$ is at most $k$ times
the probability that it has to move, which is exactly
$P[A_t \in \cM^{t-1} \setminus \cM^t] = b/|\cM^{t-1}|$ using
Lemma~\ref{lem:prob_dist}. Thus the expected cost of the algorithm is at most
$k\cdot b/|\cM^{t-1}|$, which is at most $k\cdot (-\Delta\Phi)$.

In the second case, we have $m_t < m_{t-1}$.
By Lemma~\ref{lem:Ft_size}, we know that $|\cM^t| \leq k!/m_t!$ and hence
\[
\Delta\Phi = \Phi(t)-\Phi(t-1) =  H(|\cM^t|) -H(|\cM^{t-1}|) - H(k!/m_t!)  \leq  -H(|\cM^{t-1}|) \leq -1,
\]
since
$|\cM^{t-1}| \geq 1$ and therefore $H(|\cM^{t-1}|) \geq 1$.
As the expected cost incurred by the algorithm is at most $k$, this is at most $k\cdot (-\Delta\Phi)$.
\qed

\section{Algorithm for weighted uniform metrics}
\label{sec:alg_weighted}
 In this section we prove Theorem~\ref{thm:weight_comp}. Our algorithm is a natural extension of the algorithm of Fiat and Ricklin \cite{FR94} for the weighted $k$-server problem on uniform metrics.

\paragraph{High-level idea.} 
The algorithm is defined by a recursive construction based on the following idea.
First, we can assume that the weights of the metric spaces are highly separated, i.e., $w_1 \ll w_2 \ll \dotsc \ll w_k$ (if they are not we can make them separated while losing some additional factors). So in any reasonable solution, the server $s_k$ lying in metric $M_k$ should move much less often than the other servers. 
For that reason, the algorithm moves $s_k$ only when the accumulated cost of the other $k-1$ servers reaches $w_k$. Choosing where to move $s_k$ turns out to be a crucial decision. For that reason, (in each ``level $k$-phase") during the first part of the request sequence when
the algorithm only uses $k-1$ servers, it counts how many times each point of $M_k$ is requested. We call this ``learning subphase''. Intuitively, points of $M_k$ which are requested a lot are ``good candidates'' to place $s_k$. Now, during the next $c(k)$ (to be defined later) subphases, $s_k$ visits the $c(k)$ most requested points. This way, it visits all ``important locations'' of $M_k$.
A similar strategy is repeated recursively using $k-1$ servers within each subphase.

\paragraph{Notation and Preliminaries.} We denote by $s_i^{\ALG}$ and $s_i^{\ADV}$ the server of the algorithm (resp.~adversary) that lies in metric space $M_i$. Sometimes we drop the superscript and simply use $s_i$ when the context is clear.
We set $R_k := 2^{2^{k+2}} $ and $c(k) := 2^{2^{k+1}-3}$.  Note that $c(1) = 2$ and that for all $i$,
 \begin{equation}
 \label{eq:ci_ineq}
 4 (c(i)+1) \cdot c(i) \leq 8 c(i)^2 = c(i+1). 
 \end{equation} 
Moreover, for all $i \geq 2$, we have
 \begin{equation}
 \label{eq:ri_ineq}
R_i = 8 \cdot c(i) \cdot R_{i-1}.  
 \end{equation}
 
We assume (by rounding the weights if necessary) that $w_1 = 1$ and that for $2 \leq i \leq k$, 
$w_i$ is an integral multiple of $2 (1 + c(i-1)) \cdot w_{i-1}$. Let
$m_i$ denote the ratio $w_i/(2 (1 + c(i-1)) \cdot w_{i-1})$. 

The rounding can increase the weight of each server at most by a factor of $4^{k-1} c(k-1) \cdot \dotsc \cdot c(1) \leq R_{k-1}$. So, proving a competitive ratio $R_k$ for an instance with rounded weights will imply a competitive ratio $R_k \cdot R_{k-1} < (R_k)^2$ for arbitrary weights. 

Finally, we assume that in every request $\ALG$ needs to move a server. This is without loss of generality: requests served by the algorithm without moving a server do not affect its cost and can only increase the optimal cost. This assumption will play an important role in the algorithm below.

\subsection{Algorithm Description}
The algorithm is defined recursively, where $\ALG_i$ denotes the algorithm using servers $s_1,\dotsc,s_i$. An execution of $\ALG_i$ is divided into phases. The phases are independent of each other and the overall algorithm is completely determined by describing how each phase works. We now describe the phases.

$\ALG_1$ is very simple; given any request, $\ALG_1$ moves the server to the requested point. For purposes of analysis, we divide the execution of $\ALG_1$ into phases, where each phase consists of $2(c(1)+1) = 6$ requests. 

 \vspace{2mm}

\begin{algorithm2e}[H]
\SetAlgoRefName{}
\SetAlgorithmName{Phase of $\ALG_1$}{}{}
\caption{ }
\For{$j =1$ \textbf{to} $2(c(1)+1)$} {
    Request arrives to point $p$: Move $s_1$ to $p$.
}
Terminate Phase
\end{algorithm2e}

 \vspace{2mm}

We now define a phase of $\ALG_i$ for $i \geq 2$. Each phase of $\ALG_i$ consists of exactly $c(i) + 1$ subphases. The first subphase within a phase is special and we call it the \textit{learning subphase}. During each subphase we execute $\ALG_{i-1}$ until the cost incurred is exactly $w_i$.

During the learning subphase, for each point $p \in M_i$, $\ALG_i$ maintains a count $m(p)$ of the number of requests $r$ where $p$ is requested in $M_i$, i.e. $r(i) = p$. Let us order the points of $M_i$ as $p_1, \dotsc, p_n$ such that $ m(p_1) \geq \dotsc \geq m(p_n) $ (ties are broken arbitrarily). We assume that $|M_i| \geq c(i)$ (if $M_i$ has fewer points, we add some dummy points that are never requested). Let $P$ be the set of $c(i)$ most requested points during the learning subphase, i.e. $P = \lbrace p_1 , \dotsc, p_{c(i)} \rbrace $. 

For the rest of the phase $\ALG_i$ repeats the following $c(i)$ times: it moves $s_i$ to a point $p \in P$ that it has not visited during this phase, and starts the next subphase (i.e.~it calls $\ALG_{i-1}$ until its cost reaches $w_i$).
The figure below shows the algorithm.

 \vspace{2mm}

\begin{algorithm2e}[H]
\SetAlgoRefName{}
\SetAlgorithmName{Phase of $\ALG_i$, $i \geq 2$}{}{}
\caption{\ }
Move $s_i$ to an arbitrary point of $M_i$\;
Run $\ALG_{i-1}$ until cost incurred equals $w_i$ 
\tcp*{Learning subphase}   
For $p \in M_i $, $m(p) \gets \#$ of requests such that $r(i) = p$\tcp*{Assume $ m(p_1) \geq \dotsc \geq m(p_n) $}
$P \gets  \lbrace p_1 , \dotsc, p_{c(i)}  \rbrace $\; 
\For{$j =1$ \textbf{to} $c(i)$} {
    Move $s_i$ to an arbitrary point $p \in P$\;
    $P \gets P - p$\;
	Run $\ALG_{i-1}$ until cost incurred equals $w_i$
	\tcp*{$(j+1)$th subphase}     
}
Terminate Phase
\end{algorithm2e}

 \vspace{2mm}

\subsection{Analysis}

We first note some basic properties that follow directly by the construction of the algorithm. Call a phase of $\ALG_i$, $i \geq 2$ \textit{complete}, if all its subphases are finished. Similarly, a phase of $\ALG_1$ is complete if it served exactly 6 requests.

\begin{observation}
\label{obs:sub} For $i\geq 2$,
 a complete phase of $\ALG_i$ consists of $(c(i)+1)$ subphases.
\end{observation}
\begin{observation}
\label{obs:cost_sub}
For $i\geq 2$, the cost incurred to serve all the requests of a subphase of $\ALG_i$ is $w_i$.
\end{observation}

These observations give the following corollary.
\begin{corol}
\label{cor:cost_phase}
For $i\geq 1$, the cost incurred by $\ALG_i$ to serve requests of a phase is $2 (c(i)+1) w_i$.
\end{corol}
\begin{proof}
For $i=1$ this holds by definition of the phase. For $i \geq 2$, a phase consists of $(c(i)+1)$ subphases. Before each subphase $\ALG_i$ moves server $s_i$, which costs $w_i$, and moreover $\ALG_{i-1}$ also incurs cost $w_i$.
\end{proof}
Using this, we get the following two simple properties. 

\begin{lemma}
\label{lem:basic-prop}
By definition of $\ALG$, the following properties hold:
\begin{enumerate}
\item A subphase of $\ALG_i$, $i \geq 2$, consists of $m_i$ complete phases of $\ALG_{i-1}$.
\item All complete phases of $\ALG_i$, $i \geq 1$, consist of the same number of requests.
\end{enumerate}
\end{lemma}
\begin{proof}
 The first property uses the rounding of the weights. By Corollary \ref{cor:cost_phase},
each phase of $\ALG_{i-1}$ costs $2 (c(i-1)+1) w_{i-1}$ and, in each subphase of $\ALG_i$, the cost incurred by $\ALG_{i-1}$ is $w_i$. So there are exactly $w_i/(2 (c(i-1)+1) w_{i-1}) = m_i$ phases of $\ALG_{i-1}$. 

The property above, combined with Observation~\ref{obs:sub} implies that a complete phase of $\ALG_i$ contains $m_i  \cdot(c(i) +1)$ complete phases $\ALG_{i-1}$.  Now, the second property follows directly by induction: each phase of $\ALG_1$ consists of $2(c(1)+1) = 6$ requests, and each phase of $\ALG_i$ consists of $m_i (c(i) +1)$ phases of $\ALG_{i-1}$. 
\end{proof}

Consider a phase of $\ALG_i$. The next lemma shows that, for any point $p \in M_i$, there exists a subphase where it is not requested too many times. This crucially uses the assumption that $\ALG_i$ has to move a server in every request.

\begin{lemma}
\label{lem:point_frac}
Consider a complete phase of $\ALG_i$, $i \geq 2$. For any point $p \in M_i$, there exists a subphase such that at most $1/c(i)$ fraction of the requests have $r(i) = p$.
\end{lemma} 

\begin{proof}
Let $P$ be the set of $c(i)$ most requested points of $M_i$ during the learning subphase. We consider two cases: if $p \in P$, there exists a subphase where $s_i^{\ALG}$ is located at $p$. During this subphase there are no requests such that $r(i) = p$, by our assumption that the algorithm moves some server at every request. 
Otherwise, if $p \notin P$, then during the learning subphase, the fraction of requests such that $r(i) = p$ is no more than $1 / c(i)$.
\end{proof}

To prove the competitiveness of $\ALG_k$ with respect to the optimal offline solution $\ADV_k$, the proof uses a subtle induction on $k$. 
Clearly, one cannot compare $\ALG_i$, for $i <k$ against $\ADV_k$, since the latter has more servers and its cost could be arbitrarily lower. So the idea is to compare $\ALG_i$ against $\ADV_i$, an adversary with servers $s_1,\dotsc,s_i$, while ensuring that $\ADV_i$ is an accurate estimate of $\ADV_k$ during time intervals when $\ALG_i$ is called by $\ALG_k$. To achieve this, the inductive hypothesis is required to satisfy certain properties described below. For a fixed phase, let $\cost(\ALG_i)$ and $\cost(\ADV_i)$ denote the cost of $\ALG_i$ and $\ADV_i$ respectively. 

\begin{enumerate}[(i)]
\item \textbf{Initial Configuration of $\ADV_i$.} Algorithm $\ALG_i$ (for $i<k$), is called several times during a phase of $\ALG_k$. 
As we don't know the current configuration of $\ADV_i$ each time $\ALG_i$ is called, we require that for every complete phase, $\cost(\ALG_i) \leq R_i \cdot \cost(\ADV_i)$, for any initial configuration of $\ADV_i$.

\item \textbf{Adversary can ignore a fraction of requests.} During a phase of $\ALG_i$, $\ADV_k$ may serve requests with servers $s_{i+1},\dotsc,s_k$, and hence the competitive ratio of $\ALG_i$ against $\ADV_i$ may not give any meaningful guarantee. To get around this, we will require that $\cost(\ALG_i) \leq R_i \cdot \cost(\ADV_i)$, even if the $\ADV_i$ ignores an $f(i):=4 / c(i+1)$ fraction of requests. This will allow us to use the inductive hypothesis for the phases of $\ALG_i$ where $\ADV_k$ uses servers $s_{i+1},\dotsc,s_k$ to serve at most $f(i)$ fraction of requests.

\end{enumerate}

  For a fixed phase, we say that $\ALG_i$ is strictly $R_i$-competitive against $\ADV_i$, if $\cost(\ALG_i) \leq R_i \cdot \cost(\ADV_i)$. The key result is the following.

\begin{theorem}
\label{thm:comp_ratio}
Consider a complete phase of $\ALG_i$. Let $\ADV_i$ be an adversary with $i$ servers that is allowed to choose any initial configuration and to ignore any $4/c(i+1)$ fraction of requests. Then, $\ALG_i$ is strictly $R_i$-competitive against $\ADV_i$.
\end{theorem}

Before proving this, let us note that this directly implies Theorem~\ref{thm:weight_comp}.
Indeed, for any request sequence $\sigma$, all phases except possibly the last one, are complete, so $\cost(\ALG_k) \leq R_k \cdot \cost(\ADV_k)$. The cost of $\ALG_k$ for the last phase, is at most $2 (c(k) + 1) w_k$, which is a fixed additive term independent of the length of $\sigma$. So, $\ALG_k(\sigma) \leq R_k \cdot \ADV_k(\sigma) + 2 (c(k) + 1) w_k $, and $\ALG_k$ is $R_k$-competitive. 
Together with loss in rounding the weights, this gives a competitive ratio of at mot $(R_k)^2 \leq 2^{2^{k+3}}$ for arbitrary weights.

\vspace{2mm}

We now prove Theorem~\ref{thm:comp_ratio}.
\begin{proof}[Proof of Theorem~\ref{thm:comp_ratio}]
 We prove the theorem by induction on $k$.

\vspace{2mm}

\par \textit{Base case ($i=1$):} As $R_1 > 6 $ and $ 4 / c(2) = 1/8 \leq 1/3$, it suffices to show here 
that $\ALG_1$ is strictly 6-competitive in a phase where $\ADV_1$ can ignore at most $1/3$ fraction of requests, for any starting point of $s_1^{\ADV_1}$. 
 By Lemma~\ref{lem:basic-prop}, we have $\cost(\ALG_1) =2 (c(1) + 1) = 6$. We show that $\cost(\ADV_1) \geq  1$. Consider two consecutive requests $r_{t-1},r_t$. By our assumption that $\ALG_1$ has to move its server in every request, it must be that $r_{t-1} \neq r_t$.
So, for any $t$ if $\ADV_1$ does not ignore both $r_{t-1}$ and $r_t$, then it must pay 1 to serve $r_t$. 
Moreover, as the adverary can chose the initial server location, it may (only) serve the first request at zero cost.  
As a phase consists of $6$ requests, $\ADV_i$ can ignore at most $6/3=2$ of them, so there are at most 4 requests that are either ignored or appear immediately after an ignored request. So among requests $r_2,\dotsc,r_6$, there is at least one request $r_t$, such that both $r_{t-1}$ and $r_t$ are not ignored. 

\vspace{2mm} 

\par \textit{Inductive step:} Assume inductively that $\ALG_{i-1}$ is strictly $R_{i-1}$-competitive against any adversary with $i-1$ servers that can ignore up to $4/c(i)$ fraction of requests. 

Let us consider some phase at level $i$, and let $I$ denote the set of requests that $\ADV_i$ chooses to ignore during the phase. 
We will show that $\cost(\ADV_i)\geq w_i / (2 R_{i-1})$. 
This implies the theorem, as $\cost(\ALG_i)=2 (c(i)+1) w_i$ by Lemma~\ref{lem:basic-prop} and hence,
\[
\frac{\cost(\ALG_i)}{\cost(\ADV_i)} \leq  \frac{2 (c(i)+1) w_i}{w_i / (2 R_{i-1})} = 4 (c(i)+1) R_{i-1} \leq 8 \cdot c(i) \cdot R_{i-1} =  R_i.\]

First, if $\ADV_i$ moves server $s_i$ during the phase, its cost is already at least $w_i$ and hence more than $w_i / (2 R_{i-1}) $. 
So we can assume that $s_i^{\ADV}$ stays fixed at some point $p \in M_i$ during the entire phase. So, $\ADV_i$ is an adversary that uses $i-1$ servers and can ignore all requests with $r(i) = p$ and the requests of $I$. We will show that there is a subphase where $\cost(\ADV_i) \geq w_i/(2 R_{i-1})$. 

By Lemma~\ref{lem:point_frac}, there exists a subphase, call it $j$, such that at most $1/c(i)$ fraction of the requests have $r(i) = p$. As all $c(i) +1$ subphases have the same number of requests (by Lemma~\ref{lem:basic-prop}),
even if all the requests of $I$ belong to subphase $j$, they make up at most $(4 \cdot (c(i)+1))/c(i+1) \leq 1/c(i)$ fraction of its requests, where the inequality follows from equation~\eqref{eq:ci_ineq}.  So overall during subphase $j$, $\ADV_i$ uses servers $s_1,\dotsc,s_{i-1}$ and ignores at most $2/c(i)$ fraction of requests.

We now apply the inductive hypothesis together with an averaging argument. As subphase $j$ consists of $m_i$ phases of $\ALG_{i-1}$, all of equal length, and $\ADV_i$ ignores at most $2/c(i)$ fraction of requests of the subphase, there are at most $m_i/2$ phases of $\ALG_{i-1}$ where it can ignore more than $4/c(i)$ fraction of requests. So, for at least $m_i/2$ phases of $\ALG_{i-1}$, $\ADV_i$ uses $i-1$ servers and ignores no more than $4/c(i)$ fraction of requests. By the inductive hypothesis, $\ALG_{i-1}$ is strictly $R_{i-1}$-competitive against $\ADV_i$ in these phases. As the cost of $\ALG_{i-1}$ for each phase is the same (by Lemma~\ref{lem:basic-prop}), overall $\ALG_i$ is strictly $2 R_{i-1}$ competitive during subphase $j$. As the cost of $\ALG_i$ during subphase $j$ is $w_i$, we get that $\cost(\ADV_i)\geq  w_i / 2 R_{i-1}$, as claimed.
\end{proof}

\section*{Acknowledgments}
We would like to thank Ren\'e Sitters for useful discussions on the generalized $k$-server problem.

\bibliographystyle{plain}
{\small \bibliography{references_gen_server} }

\begin{thebibliography}{10}

\bibitem{ACN00}
Dimitris Achlioptas, Marek Chrobak, and John Noga.
\newblock Competitive analysis of randomized paging algorithms.
\newblock {\em Theor. Comput. Sci.}, 234(1-2):203--218, 2000.

\bibitem{AG10}
John Augustine and Nick Gravin.
\newblock On the continuous {CNN} problem.
\newblock In {\em {ISAAC}}, pages 254--265, 2010.

\bibitem{BBMN15}
Nikhil Bansal, Niv Buchbinder, Aleksander Madry, and Joseph Naor.
\newblock A polylogarithmic-competitive algorithm for the \emph{k}-server
  problem.
\newblock {\em J. {ACM}}, 62(5):40, 2015.

\bibitem{BEK17}
Nikhil Bansal, Marek Eli{\'{a}}\v{s}, and Grigorios Koumoutsos.
\newblock Weighted k-server bounds via combinatorial dichotomies.
\newblock {\em CoRR}, abs/1704.03318, To appear in FOCS'17.

\bibitem{BBM06}
Yair Bartal, B{\'{e}}la Bollob{\'{a}}s, and Manor Mendel.
\newblock Ramsey-type theorems for metric spaces with applications to online
  problems.
\newblock {\em J. Comput. Syst. Sci.}, 72(5):890--921, 2006.

\bibitem{BEY98}
Allan Borodin and Ran El{-}Yaniv.
\newblock {\em Online computation and competitive analysis}.
\newblock Cambridge University Press, 1998.

\bibitem{BLS92}
Allan Borodin, Nathan Linial, and Michael~E. Saks.
\newblock An optimal on-line algorithm for metrical task system.
\newblock {\em J. {ACM}}, 39(4):745--763, 1992.

\bibitem{Chip-pc}
Ashish Chiplunkar.
\newblock Personal Communication. Oct 2016.

\bibitem{Chr03}
Marek Chrobak.
\newblock {SIGACT} news online algorithms column 1.
\newblock {\em {SIGACT} News}, 34(4):68--77, 2003.

\bibitem{CKPV91}
Marek Chrobak, Howard~J. Karloff, Thomas~H. Payne, and Sundar Vishwanathan.
\newblock New results on server problems.
\newblock {\em {SIAM} J. Discrete Math.}, 4(2):172--181, 1991.

\bibitem{CL91}
Marek Chrobak and Lawrence~L. Larmore.
\newblock An optimal on-line algorithm for k-servers on trees.
\newblock {\em {SIAM} J. Comput.}, 20(1):144--148, 1991.

\bibitem{KakeyaDvir}
Z.~Dvir.
\newblock On the size of {K}akeya sets in finite fields.
\newblock {\em J. Amer. Math. Soc.}, 22:1093--1097, 2009.

\bibitem{Capset}
J.~S. {Ellenberg} and D.~{Gijswijt}.
\newblock {On large subsets of $F_q^n$ with no three-term arithmetic
  progression}.
\newblock {\em ArXiv e-prints}, arXiv:1605.09223, 2016.

\bibitem{FKL+91}
Amos Fiat, Richard~M. Karp, Michael Luby, Lyle~A. McGeoch, Daniel~Dominic
  Sleator, and Neal~E. Young.
\newblock Competitive paging algorithms.
\newblock {\em J. Algorithms}, 12(4):685--699, 1991.

\bibitem{FR94}
Amos Fiat and Moty Ricklin.
\newblock Competitive algorithms for the weighted server problem.
\newblock {\em Theor. Comput. Sci.}, 130(1):85--99, 1994.

\bibitem{Guth}
L.~Guth.
\newblock {\em Polynomial {M}ethods in {C}ombinatorics}.
\newblock University Lecture Series. American Mathematical Society, 2016.

\bibitem{iw01}
Kazuo Iwama and Kouki Yonezawa.
\newblock Axis-bound {CNN} problem.
\newblock {\em IEICE TRANS}, pages 1--8, 2001.

\bibitem{IY04}
Kazuo Iwama and Kouki Yonezawa.
\newblock The orthogonal {CNN} problem.
\newblock {\em Inf. Process. Lett.}, 90(3):115--120, 2004.

\bibitem{Jukna}
Stasys Jukna.
\newblock {\em Extremal Combinatorics - With Applications in Computer Science}.
\newblock Texts in Theoretical Computer Science. An {EATCS} Series. Springer,
  2011.

\bibitem{Kou09}
Elias Koutsoupias.
\newblock The k-server problem.
\newblock {\em Computer Science Review}, 3(2):105--118, 2009.

\bibitem{KP95}
Elias Koutsoupias and Christos~H. Papadimitriou.
\newblock On the k-server conjecture.
\newblock {\em J. {ACM}}, 42(5):971--983, 1995.

\bibitem{KP96}
Elias Koutsoupias and Christos~H. Papadimitriou.
\newblock The 2-evader problem.
\newblock {\em Inf. Process. Lett.}, 57(5):249--252, 1996.

\bibitem{KT04}
Elias Koutsoupias and David~Scot Taylor.
\newblock The {CNN} problem and other k-server variants.
\newblock {\em Theor. Comput. Sci.}, 324(2-3):347--359, 2004.

\bibitem{MMS90}
Mark~S. Manasse, Lyle~A. McGeoch, and Daniel~D. Sleator.
\newblock Competitive algorithms for server problems.
\newblock {\em J. {ACM}}, 11(2):208--230, 1990.

\bibitem{Matousek}
Ji\v{r}\'{i} Matou\v{s}ek.
\newblock {\em Thirty-three Miniatures: Mathematical and Algorithmic
  Applications of Linear Algebra}.
\newblock American Mathematical Society, 2010.

\bibitem{MS91}
Lyle~A. McGeoch and Daniel~Dominic Sleator.
\newblock A strongly competitive randomized paging algorithm.
\newblock {\em Algorithmica}, 6(6):816--825, 1991.

\bibitem{Sit14}
Ren{\'{e}} Sitters.
\newblock The generalized work function algorithm is competitive for the
  generalized 2-server problem.
\newblock {\em {SIAM} J. Comput.}, 43(1):96--125, 2014.

\bibitem{SSP03}
Ren{\'{e}} Sitters, Leen Stougie, and Willem de~Paepe.
\newblock A competitive algorithm for the general 2-server problem.
\newblock In {\em {ICALP}}, pages 624--636, 2003.

\bibitem{SS06}
Ren{\'{e}}~A. Sitters and Leen Stougie.
\newblock The generalized two-server problem.
\newblock {\em J. {ACM}}, 53(3):437--458, 2006.

\bibitem{ST85}
Daniel~Dominic Sleator and Robert~Endre Tarjan.
\newblock Amortized efficiency of list update and paging rules.
\newblock {\em Commun. {ACM}}, 28(2):202--208, 1985.

\bibitem{Geffen17}
Bas van Geffen, Bart Jansen, Noud de~Kroon, Rolf Morel, and Jesper Nederlof.
\newblock Optimal algorithms on graphs of bounded width (and degree): Cutwidth
  sometimes beats treewidth, but planarity does not help.
\newblock Unpublished.

\end{thebibliography}
\newpage
\appendix

\section{Lower Bounds}
\label{sec:lower-bounds}
We present simple lower bounds on the competitive ratio of deterministic and randomized algorithms for the generalized $k$-server problem in uniform metrics. 

\paragraph{Deterministic Algorithms.}  
We show a simple construction due to \cite{KT04} that directly implies a $(2^k -
1) / k$ lower bound on the competitive ratio of deterministic algorithm.
Using a more careful argument, \cite{KT04} also improve this to $2^k-1$.

Assume that each metric space $M_i$ has $n=2$ points, labeled by $0$,$1$. A configuration of servers is a vector $c \in \lbrace 0,1 \rbrace^k$, so there are $2^k$ possible configurations. Now, a request $r =  (r_1,\dotsc,r_k)$ is unsatisfied if and only if the algorithm is in the antipodal configuration $\bar{r} = (1-r_1,\dotsc,1 -r_k)$. Let $\ALG$ be any online algorithm and $\ADV$ be the adversary. Initially, $\ALG$ and $\ADV$ are in the same configuration. At each time step, if the current configuration of $\ALG$ is $a = (a_1,\dotsc,a_k)$, the adversary requests $\bar{a}$ until $\ALG$ visits every configuration. If $p$ is the configuration that $\ALG$ visits last, the adversary can simply move to $p$ at the beginning, paying at most $k$, and satisfy all requests until $\ALG$ moves to $p$. On the other hand, $\ALG$ pays at least $2^k-1$ until it reaches $p$. Once $\ALG$ and $\ADV$ are in the same configuration, the strategy repeats.
 
\paragraph{Randomized Algorithms.} Viewing generalized $k$-server as a metrical service system (MSS), we can get a non-trivial lower bound for randomized algorithms. In particular, we can apply the $\Omega(\frac{\log N}{\log^2 \log N})$ lower bound due to Bartal et al.~\cite{BBM06} on the competitive ratio of any randomized online algorithm against oblivious adversaries, for any metrical task system on $N$ states. 
Of course, the MSS corresponding to a generalized $k$-server instance is restricted as the cost vectors may not be completely arbitrary. 
However, we consider the case where all metrics $M_i$ have $n=2$ points.
Let $s$ be an arbitrary state among the $N=2^k$ possible states.
A request in the antipodal point $\overline{s}$ only penalizes $s$
and has cost $0$ for every other state.
So the space of cost vectors here is rich enough to simulate any MSS on these
$N$ states\footnote{Note that if there is a general MSS request that has
infinite cost on some subset $S$ of states, then decomposing this into $|S|$
sequential requests where each of them penalizes exactly one state of $S$, can
only make the competitive ratio worse.}.

This implies a $\Omega (\frac{k}{\log^2 k}) $  lower bound for generalized $k$-server problem on uniform metrics. 

\end{document}